\documentclass[a4paper,UKenglish]{lipics-v2019}
\usepackage{amsfonts}
\usepackage{ifthen}
\usepackage{tikz}
\usetikzlibrary{calc}
\usetikzlibrary{math}
\usetikzlibrary{decorations}
\usepgflibrary{decorations.pathmorphing}
\usepgflibrary{arrows.meta}
\usetikzlibrary{circuits.ee.IEC}
\usepgflibrary{shapes.misc}
\usetikzlibrary{positioning}
\usepackage[linesnumbered,vlined]{algorithm2e}

\newcommand{\R}{\mathbb{R}}
\newcommand{\Rp}{\R_{\geqslant 0}}

\newcommand{\N}{\mathbb{N}}
\newcommand{\Cx}{\mathbb{C}}
\newcommand{\Q}{\mathbb{Q}}
\newcommand{\Qbar}{{\mathbb{A}}}

\newcommand{\Z}{\mathbb{Z}}
\newcommand{\K}{\mathbb{K}}

\newcommand{\set}[1]{\left\{#1\right\}}

\DeclareMathOperator{\GL}{GL}

\DeclareMathOperator{\diag}{diag}

\newcommand{\Zcl}[2][]{\overline{#2}^{#1}}

\newcommand{\Gen}[1]{\left\langle#1\right\rangle}

\allowdisplaybreaks

\newcommand{\ie}{i.e.,}

\SetAlgoNoEnd
\DontPrintSemicolon
\SetKw{KwFrom}{from}
\SetKw{KwAnd}{and}
\SetKwInOut{Input}{input}
\SetKwInOut{Output}{output}

\title{Algebraic Invariants for Linear Hybrid Automata} 

\author{Rupak Majumdar}{Max Planck Institute for Software Systems, Germany}{%
rupak@mpi-sws.org}{}{}

\author{Jo\"el Ouaknine}{Max Planck Institute for Software Systems, Germany\and%
Department of Computer Science, Oxford University, UK}{%
joel@mpi-sws.org}{}{ERC grant AVS-ISS (648701) and Deutsche
Forschungsgemeinschaft (DFG, German Research Foundation) ---
Projektnummer 389792660 --- TRR 248}

\author{Amaury Pouly}{CNRS, IRIF, Université Paris Diderot, France}{%
amaury.pouly@irif.fr}{https://orcid.org/0000-0002-2549-951X}{Part of this work was done at MPI-SWS}

\author{James Worrell}{Department of Computer Science, Oxford University, UK}{%
jbw@mpi-sws.org}{}{EPSRC Fellowship EP/N008197/1}

\authorrunning{R. Majumdar, J. Ouaknine, A. Pouly and J. Worrell}

\Copyright{Joël Ouaknine, Amaury Pouly and James Worrell}

\ccsdesc[100]{Theory of computation~Models of computation}
\ccsdesc[100]{timed and
  hybrid models}

\keywords{Hybrid automata, algebraic invariants}

\category{}

\relatedversion{}

\supplement{}

\acknowledgements{}

\nolinenumbers 

\EventEditors{John Q. Open and Joan R. Access}
\EventNoEds{2}
\EventLongTitle{42nd Conference on Very Important Topics (CVIT 2016)}
\EventShortTitle{CVIT 2016}
\EventAcronym{CVIT}
\EventYear{2016}
\EventDate{December 24--27, 2016}
\EventLocation{Little Whinging, United Kingdom}
\EventLogo{}
\SeriesVolume{42}
\ArticleNo{23}

\begin{document}

\maketitle

\begin{abstract}
  We exhibit an algorithm to compute the strongest algebraic (or
  polynomial) invariants that hold at each location of a given
  unguarded linear hybrid automaton (i.e., a hybrid automaton
  having only unguarded transitions, all
  of whose assignments are given by affine expressions, and all of
  whose continuous dynamics are given by linear differential
  equations).  Our main tool is a control-theoretic result of
  independent interest: given such a linear hybrid automaton, we show
  how to discretise the continuous dynamics in such a way that the
  resulting automaton has precisely the same algebraic invariants.
\end{abstract}

\section{Introduction}

Invariants are one of the most fundamental and useful
notions in the quantitative sciences, appearing in a wide range of
contexts, from gauge theory, dynamical systems, and control theory
in physics, mathematics, and engineering to program verification,
static analysis, abstract interpretation, and programming language
semantics (among others) in computer science. In spite of decades of
scientific work and progress, automated invariant synthesis remains
a topic of active research, particularly in the fields of computer-aided
verification and program analysis, and plays a central role in methods
and tools seeking to establish correctness properties of computer
systems; see, e.g., \cite{KCBR18}, and particularly Sec.~8 therein.

In this paper, we consider the task of computing \emph{strongest
  algebraic inductive invariants} for \emph{unguarded linear
  hybrid automata}. Hybrid automata are a formalism for describing
systems or processes that combine discrete and continuous evolutions
over their state variables. A hybrid automaton is therefore equipped
with a finite set of real-valued variables, as well as a finite set of
control location (or modes). In each location, the variables evolve according to
some differential dynamics. Transitions between control locations may
effect discrete updates (also known as \emph{resets}) to these
variables. A hybrid automaton is \emph{unguarded} if the
transitions do not have any guards, or preconditions, in order to be
fired, and it is \emph{linear} if the discrete updates on the
variables consist entirely of affine transformations, and the
continuous dynamics within each control location are defined by
linear differential equations.

An \emph{invariant} assigns to each control location a fixed set of real
values in such a way that through any trajectory of the hybrid
automaton, the values of the variables always remain within the
invariant. The invariant is \emph{inductive} provided, informally
speaking, that it is itself preserved by the (continuous and discrete)
dynamics of the hybrid automaton. Finally, an invariant is
\emph{algebraic} (or polynomial) if it consists in a collection of varieties (or
algebraic sets), i.e., positive Boolean combination of polynomial
equalities. As it happens, the strongest polynomial invariant (i.e.,
smallest variety with respect to set inclusion) is obtained by taking
the Zariski closure of the set of reachable configurations in each
control location; such an invariant is always inductive provided that
the dynamics are Zariski continuous.

There is a rich history of research into the computation of algebraic
invariants for various classes of (discrete) computer programs; we
refer the reader to our recent paper~\cite{HOP018} and references
therein. There has also been a substantial amount of work on algebraic
invariant generation for hybrid systems, albeit in more recent years.
One of the earliest pieces of work on this topic is by
Rodr\'igez-Carbonell and Tiwari~\cite{CE05}, who consider linear
dynamical systems (i.e., linear hybrid automata with a single discrete
location and no transition) and show how to compute strongest
algebraic invariants for these. They then leverage
abstract-interpretation techniques to derive algebraic invariants for
linear hybrid automata, however without guarantees on the strength of
the invariants. In \cite{SSM08}, Sankaranarayanan et al.\ compute
algebraic invariants for polynomial hybrid systems directly using
constraint solving over template invariants (without however
guaranteeing to obtain the strongest invariant). In subsequent work,
Sankaranarayanan shows how to compute strongest algebraic invariants
up to a fixed degree~\cite{San10} for the same class of
automata. Using different analytic techniques, Ghorbal and Platzer
show in~\cite{GP14} how to compute algebraic invariants and
differential invariants for polynomial hybrid automata; more
precisely, they show that deciding whether a collection of algebraic
sets forms an algebraic invariant is decidable; they do not however
provide a procedure to guarantee that a given invariant is the
strongest possible.

\textbf{Main results.} 
Our contributions in the present paper are threefold. \textbf{First,} for the
class of unguarded linear hybrid automata, building on our
recently developed invariant-generation techniques for affine
programs~\cite{HOP018}, we show how to compute strongest algebraic
invariants. Our main technical tools come from linear algebra, algebraic
geometry, and Diophantine geometry.
\textbf{Second,} we show how one can discretise an unguarded
linear hybrid automaton in such a way that the discretised version has
precisely the same algebraic invariants as the original one; we are
not aware of any such result in the extant control-theoretic and
cyber-physical systems literature. \textbf{And
third,} we show that as soon as equality guards are allowed, even for
the restricted class of linear \emph{switching systems}, there cannot
exist an algorithm for computing strongest algebraic invariants,
thereby establishing clearly a hard theoretical limit on how far the
work presented here can be extended.

We now provide a sightly more detailed overview of our approach and results.
In Section~\ref{sec:closure_lin_diff_eq} we consider a simple class
of hybrid systems, with
purely continuous dynamics, called switching systems.  A switching
system can transition arbitrarily between modes, but the variables are
not reset when changing mode.  It is natural to think of mode switches
as being determined by an external controller which provides inputs to
the system.  We show that for a switching system, the Zariski closure
of the set of reachable configurations is an irreducible variety.  We
exploit this feature to give a conceptually simple algorithm to
compute the strongest algebraic invariant of a given switching system.

In the rest of the paper we develop this basic result in two
directions.  In Sections~\ref{sec:const-discretisation}
and~\ref{sec:fin-discretisation} we generalise the above analysis to
accommodate discrete transitions that may reset variables (e.g., the
bouncing ball in Example~\ref{sec:bouncing_ball}, in which the
velocity is reset on hitting a boundary, or the RC circuit in
Example~\ref{sec:RC_circuit}, in which the electrical currents and
voltages change abruptly when a switch is turned on or off).  In this
setting the dynamics is not necessarily continuous.  In fact, our
approach here is to eliminate the continuous dynamics altogether by
introducing a discretisation construction on hybrid automata that
preserves all algebraic invariants.  More formally, given a hybrid
automaton we show how to construct a finite (discrete-time) affine
program that has the same same set of variables and the same algebraic
invariants as the original hybrid automaton.  We can then rely on a
result in~\cite{HOP018} to compute strongest algebraic invariants.

In another direction, in Section~\ref{sec:guards} we consider the
problem of computing the strongest algebraic invariant for switching
systems that are augmented with the ability to test variables for zero
on transitions.  Here again the dynamics are exclusively continuous.
We show that it is undecidable in general to compute a strongest
algebraic invariant for such systems.  Roughly speaking, we prove this
result by defining a simulation of an arbitrary Minsky machine
$\mathcal{C}$ by a hybrid automaton $\mathcal{A}$, such that we can
effectively determine whether the set of reachable configurations of
$\mathcal{C}$ is infinite from the strongest algebraic invariant of
$\mathcal{A}$.

\section{Examples}

\subsection{Bouncing ball}\label{sec:bouncing_ball}

Consider a ball that bounces on horizontal slabs, as illustrated in Figure~\ref{fig:bouncing_ball}.
The ball is moving at constant horizontal speed $c$ and is subject to gravity along the vertical
axis. We assume that there is no friction and that collisions are perfectly elastic. We do
not want to make any assumption on the location of the slabs, to obtain the most general system.
Thus from the point of the view of system, the positions of the slabs are `nondeterministic'
and the slabs can `appear' at any moment.

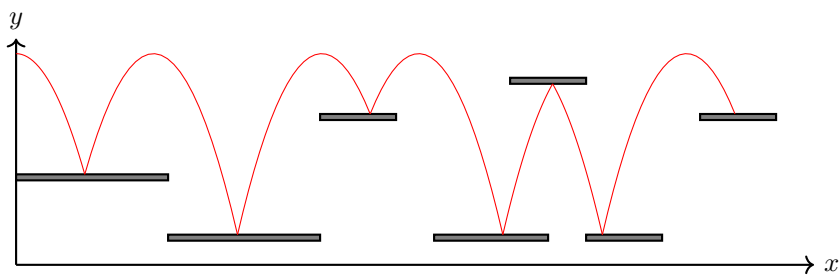
\begin{figure}[h]
    \begin{center}
    \begin{tikzpicture}[
            my slab/.style = {thick,fill=gray},
            my axis/.style = {thick, ->},
            my ball curve/.style = {red},
            yscale=0.4
        ]
        \tikzmath{
            \g = 9.8;
            \c = 1;
            \t = 0;
            \x = 0;
            \h = 7;
            \y = \h;
            \vy = 0;
        }
        \draw[my axis] (0,0) -- (10.5,0) node[right] {$x$};
        \draw[my axis] (0,0) -- (0, \h+0.5) node[above] {$y$};
        \message{DEBUG START t = \t, x = \x, y = \y, vy = \vy}
        \foreach \slabx/\slaby/\slabdx/\slabdir [remember=\t,remember=\x,remember=\y,remember=\vy]
                in {0/3/2/-1, 2/1/2/-1, 4/5/1/-1, 5.5/1/1.5/-1, 6.5/6/1/1, 7.5/1/1/-1, 9/5/1/-1} {
            \tikzmath{
                \newt = (\g*\t+\vy-\slabdir*sqrt(\vy*\vy-2*\g*\slaby+2*\g*\y))/\g;
            }
            \draw[my slab] (\slabx,\slaby) rectangle ++(\slabdx,0.2*\slabdir);
            \draw[my ball curve] plot[domain=\t:\newt,variable=\u]
                ({\x+\c*(\u-\t},{\y+(\u-\t)*\vy-1/2*(\u-\t)*(\u-\t)*\g});
            \tikzmath{
                \y = \y+(\newt-\t)*\vy-1/2*(\newt-\t)*(\newt-\t)*\g;
                \vy = \vy-(\newt-\t)*\g;
                \x = \x + \c * (\newt - \t);
                \t = \newt;
                \vy = -\vy;
            }
            \message{DEBUG t = \t, x = \x, y = \y, vy = \vy}
        }
    \end{tikzpicture}
    \end{center}
    \caption{A ball bouncing on horizontal slabs.\label{fig:bouncing_ball}}
\end{figure}

We fix an arbitrary coordinate system in which the ball starts at position $(0,h)$ with initial velocities
$(c,0)$. The system is modelled using a single discrete location. There are three constants $c$
(the horizontal speed), $h$ (the initial height), and $g$ (approximately $9.8 ms^{-2}$).
Technically speaking we could also include $m$ (the mass of the ball),
but it will not feature in the final set of invariants that we derive nor in our differential equations.
There are five variables: $t$, $x$, $y$, $v_x$, and $v_y$, where $t$ is the time
variable. The differential equations are simply Newton's laws of motion. There is a single reset
with no guards modelling the action of the ball bouncing on each of
the slabs (notably ensuring that the vertical velocity is instantly inverted).
We obtain the hybrid system described in Figure~\ref{fig:hybrid:bouncing_ball}.

Now we come to the invariants. The most obvious is that the horizontal speed is constant:
$v_x = c$, which in turn entails that $x = tc$.
Next, consider an inertial coordinate system moving horizontally to the right at speed $c$.
In that system energy must be conserved. Initially there is no kinetic energy,
and all the potential energy amounts to $mgh$ (where $m$ is the mass of the ball).
At any subsequent time $t$, the sum of the kinetic and potential energy must therefore sum to
that value, \ie{} $\tfrac{1}{2}mv_y^2 + mgy = mgh$, so $v_y^2 + 2g(y-h) = 0$.
These must be the only invariants: the ambient space is 5-dimensional (given that our variables are
$t$, $x$, $y$, $v_x$, and $v_y$), and the corresponding variety (with 3 equations) is two dimensional.
But the system has indeed exactly two degrees of freedom, since t and $v_y$ can be set to arbitrary values
(provided $|v_y| \leq gt$), and once t and $v_y$ are fixed, every other variable is fixed.
In summary, the strongest algebraic invariant is the conjunction of
the following three equations:
\[
    v_x=c,
    \qquad
    x=tc,
    \qquad
    v_y^2 + 2g(y-h) = 0.
\]

\begin{figure}
    \begin{subfigure}[t]{0.35\textwidth}
        \centering
        \begin{tikzpicture}
            \node[draw,rectangle,very thick,inner sep=3pt] (state) {
                $\begin{array}{c@{\hspace{0.7ex}}l}
                \dot{x}     &= v_x \\
                \dot{y}     &= v_y \\
                \dot{v}_x   &= 0\\
                \dot{v}_y   &= -g \\
                \dot{t}     &= 1
                \end{array}$};
            \draw[->,thick] ($(state.west)+(-2cm,0)$) -- (state)
                node[midway,above] {\begin{tabular}{c}$t:=0$\\$x:=0$\\$y:=h$\end{tabular}}
                node[midway,below]
                {\begin{tabular}{c}$v_x:=c$\\$v_y:=0$\end{tabular}};
            \draw (state) edge[thick,loop above,distance=1cm,->] node[above,midway] {$v_y:=-v_y$} (state);
            \node (invisible) at (0,-1.65) {};
        \end{tikzpicture}
        \subcaption{Hybrid system modelling a bouncing ball.}\label{fig:hybrid:bouncing_ball}
    \end{subfigure}\hfill
    \begin{subfigure}[t]{0.63\textwidth}
        \centering
        \begin{tikzpicture}
            \node[draw,rectangle,very thick,inner sep=3pt] (open) {
                $\begin{array}{c@{\hspace{0.7ex}}l}
                \dot{I}     &= 0 \\
                \dot{I}_R   &= -\frac{1}{RC}I_R\\
                \dot{V}_R   &= -\tfrac{1}{C}I_R \\
                \dot{Q}     &= I_R\\
                \dot{V}_C   &= \frac{1}{C}I_R
                \end{array}$};
            \node[above=1ex of open] {\textsc{open}};
            \node[draw,rectangle,very thick,inner sep=3pt] at (6,0) (closed) {
                $\begin{array}{c@{\hspace{0.7ex}}l}
                \dot{I}     &= -\frac{1}{RC}I_R \\
                \dot{I}_R   &= -\frac{1}{RC}I_R\\
                \dot{V}_R   &= -\tfrac{1}{C}I_R \\
                \dot{Q}     &= I_R\\
                \dot{V}_C   &= \frac{1}{C}I_R
                \end{array}$};
            \node[above=1ex of closed] {\textsc{closed}};
            \draw[thick,->] ($(open.north east)!0.9!(open.east)$) -- ($(closed.north west)!0.9!(closed.west)$)
                node[above,midway] {
                    $\begin{array}{c@{\hspace{0.3ex}}l}
                    \scriptstyle I   &:= \scriptstyle\tfrac{1}{R}(V-V_C)\\
                    \scriptstyle I_R &:= \scriptstyle\tfrac{1}{R}(V-V_C)\\
                    \scriptstyle V_R &:= \scriptstyle V-V_C
                    \end{array}$};
            \draw[thick,<-] ($(open.south east)!0.9!(open.east)$) -- ($(closed.south west)!0.9!(closed.west)$)
                node[below,midway] {
                    $\begin{array}{c@{\hspace{0.3ex}}l}
                    \scriptstyle I   &:= \scriptstyle 0 \\
                    \scriptstyle I_R &:= \scriptstyle-\tfrac{1}{R}V_C\\
                    \scriptstyle V_R &:= \scriptstyle -V_C
                    \end{array}$};
        \end{tikzpicture}
        \subcaption{Hybrid system modelling an RC circuit.}\label{fig:hybrid:rc_circuit}
    \end{subfigure}
    \caption{Examples of hybrid systems.}\label{fig:hybrid}
\end{figure}
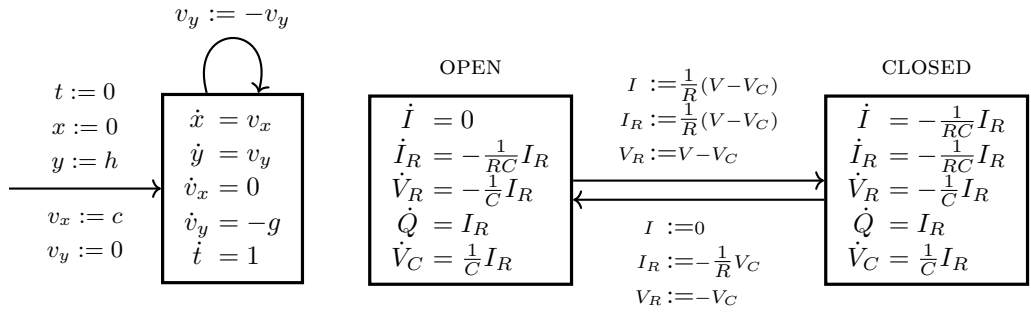

\subsection{RC circuit}\label{sec:RC_circuit}

\begin{figure}[h]
    \begin{center}
    \begin{tikzpicture}[
                circuit ee IEC,small circuit symbols,
                set resistor graphic=var resistor IEC graphic,
                set make contact graphic= var make contact IEC graphic
            ]
            \coordinate (top left);
            \coordinate (top right) at (5.5,0) {};
            \coordinate (bot left) at (0, -2) {};
            \coordinate (bot right) at (5.5, -2) {};
            \def\switchpos{0.3}
            \node[contact] (bot to switch) at ($(bot left)!\switchpos!(bot right)$) {};
            \node[contact={info=right:{\tiny OPEN}}] (top to switch)
                at ($(top left)!\switchpos!(top right)+(0,-2ex)$) {};
            \draw (top left) to[current direction={pos=0.1,info=above:$I$},
                                contact={pos=\switchpos-0.05,info=above:{\tiny CLOSED}},
                                make contact={pos=\switchpos,rotate=180},
                                current direction={pos=0.55, info=above:$I_R$},
                                resistor={pos=0.8, info=below:$R$, info=above:$V_R$}] (top right);
            \draw (top right) to[capacitor={midway,info=above:$V_R$,info=below left:$Q$,
                                info=below right:$C$}] (bot right);
            \draw (bot left) -- (bot to switch);
            \draw (bot to switch) -- (bot right);
            \draw (top left) to[battery={pos=0.4,info=below:$V$},battery={pos=0.6}] (bot left);
            \draw (bot to switch) -- (top to switch);
        \end{tikzpicture}
    \end{center}
    \caption{An RC circuit with a switch to disconnect the battery.\label{fig:rc_circuit}}
\end{figure}
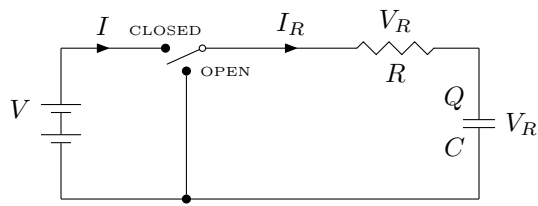

Consider an RC circuit with a switch, as illustrated in
Figure~\ref{fig:rc_circuit}. When the switch is on, the capacitor is
connected to a battery and charges. When the switch is off, the
capacitor discharges through the resistor. The battery has constant
voltage $V$, the resistor has resistance $R$ and the capacitor has
capacity $C$. There are $5$ variables: the current $I$ in the wire
between the battery and the switch, the voltages $V_R$ and $V_C$
across the resistor and capacitor respectively, the current $I_R$
flowing through the resistor, and finally the charge $Q$ held by the
capacitor.  All derivatives are with respect to time, though this time
we choose not to include an explicit variable for the passage of
time. There are two discrete locations, \textsc{open} and \textsc{closed}, 
corresponding to the two possible positions of the switch. We assume the
switch starts in the \textsc{open} position. When switching from \textsc{open} to
\textsc{closed}, all variables but $Q$ and $V_C$ experience a reset.  We obtain
the hybrid system described in Figure~\ref{fig:hybrid:rc_circuit}.

In this example, there is one set of invariants per location. In both locations, we clearly have
the invariants associated with the passive components: $Q=CV_C$ and $V_R=RI_R$.
In the \textsc{open} location, we further have the invariants $I=0$ and $V_R=-V_C$. On the hand,
in the \textsc{closed} location, we have $I=I_R$ and $V=V_R+V_C$.
These must be the only invariants: once $Q$ is fixed, all variables are uniquely determined.
In summary, the invariants are
\begin{center}
    \begin{tabular}{c@{\hspace{1cm}}llll}
        \textsc{open}: & $Q = CV_C,$ & $V_R = RI_R,$ & $I = 0,$ & $V_R = -V_C,$\\
        \textsc{closed}: & $Q = CV_C,$ & $V_R = RI_R,$ & $I = I_R,$ & $V_R = V-V_C.$
    \end{tabular}
\end{center}

\section{Mathematical Background}\label{sec:math}

Let $\K$ be a field.  Given a set $X\subseteq\K^n$, we denote by
$\mathbf{I}(X)$ the ideal of polynomials in $\K[x_1,\ldots,x_n]$ that
vanish on $X$. Given an ideal $I\subseteq\K[x_1,\ldots,x_n]$, we
denote by $V(I) \subseteq \K^n$ the set of common zeroes of the
polynomials in $I$.  A set $X \subseteq \K^n$ is said to be an
\emph{affine variety} (also called an \emph{algebraic set}) if
$X=V(I)$ for some ideal $I\subseteq\K[x_1,\ldots,x_n]$.  By the
Hilbert Basis Theorem, every affine variety can be described as the
set of common zeroes of finitely many polynomials.  We identify
$\GL_n(\K)$, the set of $n\times n$ invertible matrices with entries
in $\K$, with the variety
$\{(A,y) \in \K^{n^2+1} : \det(A) \cdot y = 1\}$.

Given an affine variety $X\subseteq\K^n$, the \emph{Zariski topology}
on $X$ has as closed sets the subvarieties of $X$, i.e., those sets
$A\subseteq X$ that are themselves affine varieties in $\K^n$.  Given
an arbitrary set $S\subseteq X$, we write $\Zcl{S}$ for its closure in
the Zariski topology on $X$.

A set $S \subseteq X$ is \emph{irreducible} if
for all closed subsets $A_1,A_2\subseteq X$ such that
$S\subseteq A_1 \cup A_2$ we have either $S\subseteq A_1$ or
$S\subseteq A_2$.  It is well known that the Zariski topology on a
variety is Noetherian.  In particular, any closed subset $A$ of $X$
can be written as a finite union of \emph{irreducible components},
where an irreducible component of $A$ is a maximal irreducible closed
subset of $A$.

The class of \emph{constructible} subsets of a variety $X$ is obtained
by taking all Boolean combinations (including complementation) of
Zariski closed subsets.  Suppose that the underlying field $\K$ is
algebraically closed.  Since the first-order theory of algebraically
closed fields admits quantifier elimination, the constructible subsets
of $X$ are exactly the subsets of $X$ that are first-order definable
over $\K$.

A function $f:\K^m\rightarrow \K^n$ is said to be a \emph{polynomial
  map} if there exist polynomials
$p_1,\ldots,p_n \in \K[x_1,\ldots,x_m]$ such that
$f(\boldsymbol{a}) = (p_1(\boldsymbol{a}),\ldots,p_n(\boldsymbol{a}))$
for all $\boldsymbol{a}\in \K^m$.  
Recall that
polynomial maps are Zariski-continuous and thus
$f(\Zcl{X})\subseteq\Zcl{f(X)}$ for a polynomial map $f$.  In particular
matrix multiplication is a Zariski-continuous map 
$\K^{n^2} \times \K^{n^2} \rightarrow \K^{n^2}$.

Given a complex variety $V\subseteq\Cx^n$, the intersection
$V\cap\R^n$, which is a real variety, can be computed
effectively. Indeed if $V$ is represented by the ideal
$I\subseteq \Cx[x_1,\ldots,x_n]$, then $V\cap\R^n$ is represented by
the ideal generated by $\set{p_\R,p_{i\R}:p\in I}$ where $p_\R$ and
$p_{i\R}$ respectively denote the real and imaginary parts of the
polynomial $p$.  Given $S \subseteq \R^n$, write $\Zcl[\R]{S}$ for its
real Zariski closure and $\Zcl{S}$ for its complex Zariski closure
(i.e., we treat the complex Zariski closure of $S \subseteq \R^n$ as
the default).  It is straightforward to verify that
$\Zcl[\R]{S} = \Zcl{S} \cap \R^n$.

\section{Algebraic Invariants for Hybrid Automata}\label{sec:alg_invariants_hybrid}

We are concerned with computing strongest algebraic invariants for
the subclass of hybrid automata that has no guards, linear discrete
updates, and linear continuous dynamics.  Each location of such an
automaton specifies a linear differential equation ($x'=Ax$), and each
transition between states (nondeterministic choices are allowed)
specifies a linear transformation ($x\mapsto Bx$).  Such a hybrid
automaton can be pictured as follows:

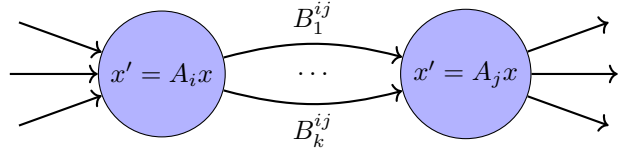
\begin{center}
\begin{tikzpicture}[
        mynode/.style={fill=blue!30!white,circle,draw},
        mynode final/.style={mynode,double, double distance=0.5mm},
        myarrow/.style={->,thick},
    ]
    \node[mynode] at (0,0) (n1) {$x'=A_ix$};
    \node[mynode] at (4,0) (n2) {$x'=A_jx$};
    \path (n1) ++(20:-2) edge[myarrow] (n1);
    \path (n1) ++(0:-2) edge[myarrow] (n1);
    \path (n1) ++(-20:-2) edge[myarrow] (n1);
    \path (n2) edge[myarrow] ++(-20:2);
    \path (n2) edge[myarrow] ++(0:2);
    \path (n2) edge[myarrow] ++(20:2);
    \draw[myarrow] (n1) edge[bend left=15] node[midway,above] {$B^{ij}_1$} (n2);
    \path (n1) edge[opacity=0] node[opacity=1,midway] {$\ldots$} (n2);
    \draw[myarrow] (n1) edge[bend right=15] node[midway,below] {$B^{ij}_k$} (n2);
\end{tikzpicture}
\end{center}
Formally, such an automaton $\mathcal{A}$ in dimension $d$ is a tuple
$(Q,A,E,T)$, where $Q$ is a finite set of locations,
$A = \{A_q:q \in Q\}$ is a family of real $d\times d$ matrices,
$E\subseteq Q\times\R^{d\times d}\times Q$ is a set of transitions
labelled by real $d\times d$ matrices, and $\{ T_q : q \in Q\}$ is a
family of algebraic subsets of $\R^d$.  Matrix $A_q\in\R^{d\times d}$
describes the continuous dynamics at location $q\in Q$ and
$T_q \subseteq \R^d$ is the set of initial states in location $q$.  We
assume that the entries of all matrices are algebraic numbers and that
the polynomials defining $T_q$ have algebraic coefficients.

We will consider the subclasses of automata with the following
restrictions:
\begin{itemize}
\item \emph{affine programs}: $A_q=0$ for all $q\in Q$, and $E$ is finite,
\item \emph{constructible affine programs}: $A_q=0$ for all $q\in\Q$, and $E$ is constructible,
\item \emph{switching systems}: $E=\{(p,I_n,q):p,q\in Q\}$, i.e., every pair of locations is connected by an edge that does not update the variables,
\item \emph{linear hybrid automata}: $E$ is finite.
\end{itemize}
In affine programs, variables are only updated on discrete edges: there
is no continuous evolution within locations.  At the other end of the
spectrum, in switching systems variables only evolve continuously, and
there are no discrete updates.  The full class of linear hybrid
automata accommodate both discrete an continuous updates to the
variables.

The \emph{collecting semantics} of $\mathcal{A}$ assigns to each
location $q\in\ Q$ the set $S_q\subseteq\R^d$ of states that can occur
in location $q$ during a run of the automaton, starting from a
configuration $(q,\boldsymbol a)$ for some $\boldsymbol a \in
T_q$. Formally, this is the smallest family (with respect to set inclusion)
such that
\[\begin{array}{rcll}
S_{q}&\supseteq&T_q &\text{for all }q\in Q,\\
S_{q}&\supseteq& BS_{p}&\text{for all }(p,B,q)\in E,\\
S_{q}&\supseteq& e^{A_qt}S_q&\text{for all }t\in\R_{\geqslant0}.\\
  \end{array}\]
Equivalently, let the operator $\Phi_\mathcal{A}:\mathcal{P}(\R^d)^Q\to\mathcal{P}(\R^d)^Q$ be defined by
\[\Phi_\mathcal{A}(S)_q=T_q\cup \bigcup_{t\geqslant0}e^{A_qt}S_q\cup\bigcup_{(p,B,q)\in E}BS_p.\]
Then $S$ is the least fixed-point of $\Phi_\mathcal{A}$ with respect 
to set inclusion.
Such a least fixed-point exists because
$\Phi_\mathcal{A}$ is monotone.

  In general we say that a family of set $\{ S'_q \}_{q \in Q}$, with
  $S'_q \subseteq \R^d$ is an \emph{inductive invariant} if
  $\Phi_{\mathcal{A}}(S') \subseteq S'$, i.e., the family is
  pre-fixed-point of $\Phi_{\mathcal{A}}$.  If each set $S'_q$ is
  algebraic then we moreover say that $\{ S'_q \}_{q \in Q}$ is an
  \emph{inductive algebraic invariant}.

Given $P\in\R[x_1,\ldots,x_d]$, we say that the relation $P=0$ holds
at location $q$ if $P$ vanishes on $S_q$. We are interested in
computing at each location $q\in\Q$ a finite set of polynomials that
generates the ideal $I_q:=\mathbf{I}(S_q)\subseteq\R[x_1,\ldots,x_d]$
of all polynomial relations that hold at location $q$.  The real
variety $V_q:=V(I_q)=\Zcl[\R]{S_q}$ corresponding to $I_q$ is the
Zariski closure of $S_q$ viewed a subset of the affine space $\R^d$.

Note that the collection $V(\mathcal{A}):=\set{V_q:q\in Q}$ defines an
inductive algebraic invariant.  Inductiveness amounts to the following
two claims:
\begin{itemize}
\item for every edge $(p,B,q)\in\ E$, we have $BV_p\subseteq V_q$,
\item for every $q\in Q$ and $t\in\Rp$, we have $e^{A_qt}V_q\subseteq V_q$.
\end{itemize}
The first point follows from the fact that $x\mapsto Bx$ is Zariski-continuous; the second point
likewise follows from the fact that for every $t\in\Rp$ the map  $x\mapsto e^{A_qt}x$
is Zariski-continuous. 

The discussion above shows that $V(\mathcal{A})$, the Zariski closure
of the collecting semantics, is the least inductive invariant of
$\mathcal{A}$.  Previously we have shown how to compute the Zariski
closure of the collecting semantics of an affine program:

\begin{theorem}[\cite{HOP018}]\label{thm:lics-main}
  There is an algorithm that given a
  constructible affine program $\mathcal{A}$ computes
  $V(\mathcal{A})=\set{V_q:q\in\Q}$---the real Zariski closure of its
  collecting semantics.
\end{theorem}
Note that this theorem also works for (finite) affine programs since
those are a particular case of constructible affine programs. 

The main result of the current paper extends the above result by
accommodating the continuous dynamics of hybrid automata:

\begin{theorem}\label{thm:main}
There is an algorithm that given a linear hybrid automaton $\mathcal{A}$
computes $\set{V_q:q\in Q}$---the real Zariski closure of
its collecting semantics.
\end{theorem}

\section{Linear Continuous Dynamics and Switching Systems}\label{sec:closure_lin_diff_eq}

The first step towards computing Zariski closure in the general case
is to be able to handle the case of one differential equation.  Write
$\Qbar$ for the field of algebraic numbers.  Let $A\in\Qbar^{d\times
  d}$ and $x_0\in\Qbar^{d\times d}$, then the solution to $x(0)=x_0$,
$x'(t)=Ax(t)$ is given by $x(t)=e^{At}x_0$. We are interested in
computing the Zariski closure of the orbit
$\Zcl{\set{x(t):t\in\Rp}}$. Since the map $\phi:M\mapsto Mx_0$ is
Zariski-continuous, we have that
\[
\Zcl{\set{x(t):t\in\Rp}}
    =\Zcl{\set{e^{At}x_0:t\in\Rp}}
    =\Zcl{\phi(\set{e^{At}:t\in\Rp})}
    =\Zcl{\phi(\Zcl{\set{e^{At}:t\in\Rp}})}
\]
and thus it suffices to compute $\Zcl{\set{e^{At}:t\in\Rp}}$.
Furthermore, let $\mathcal{O}_A:=\set{e^{At}:t\in\R}$, which is a commutative group. Then
$\Zcl{\set{e^{At}:t\in\Rp}}=\Zcl{\mathcal{O}_A}$: the left-to-right inclusion is clear.
The converse inclusion comes from the fact that an exponential polynomial (in one variable), being an analytic function,
vanishes over $\Rp$ if and only if it vanishes over $\R$. Thus it suffices
to compute $\Zcl{\mathcal{O}_A}$.

The following lemma gives a description of the ideal of the variety $\Zcl{\mathcal{O}_A}$
when $A$ is diagonal.

\begin{lemma}\label{lem:ideal_lin_diff_eq_diag}
    Let $A=\diag(\lambda_1,\ldots,\lambda_d)\in\Qbar^{d\times d}$ be a diagonal matrix,
    then \[\Zcl{\mathcal{O}_A}=\set{\diag(z_1,\ldots,z_d):\forall p\in I,p(z_1,\ldots,z_d)=0}\, , \]
    where $I=\left\langle z^a-z^b:a-b\in L\right\rangle$ and $L=\set{n\in\Z^d:n_1\lambda_1+\cdots+n_d\lambda_d=0}$.
    Furthermore, one can compute a basis for $L$ considered as an abelian group under addition.
\end{lemma}
\begin{proof}
Clearly $e^{At}=\diag(e^{\lambda_1t},\ldots,e^{\lambda_dt})$. Since the set of diagonal matrices is closed,
then the closure is of the form \[ \Zcl{\mathcal{O}_A}=\set{\diag(z_1,\ldots,z_n):p_1(z)=\cdots=p_k(z)=0}\]
for some polynomials $p_1,\ldots,p_k$. Thus, the ideal $I$ of the closure is generated by all polynomials
$p$ such that $p(e^{\lambda_1t},\ldots,e^{\lambda_dt})=0$ for all $t\in\R$. Let $J$ be the ideal of
all polynomials $x^a-x^b$ with $a,b\in\N^d$ such that $\lambda\cdot a=\lambda\cdot b$.
Clearly $J\subseteq I$ since if $\lambda\cdot a=\lambda\cdot b$, then
$(e^{\lambda_1 t})^{a_1}\cdots(e^{\lambda_d t})^{a_d}-(e^{\lambda_1 t})^{b_1}\cdots(e^{\lambda_d t})^{b_d}=e^{(\lambda\cdot a)t}-e^{(\lambda\cdot b)t}=0$ for all $t\in\R$.
Conversely, assume by contradiction that $p\in I\setminus J$ and write
$p=\sum_{i=1}^rb_im_i$
for some $b_i\in\Qbar$ and monomials $m_1,\ldots,m_r$. Further choose $p$ so that $r$ is minimal.
For each monomial $m_i(x)=x_1^{a_1}\cdots x_d^{a_d}$, let $\mu_i:=\lambda\cdot a$. Then we must have
$\mu_i\neq\mu_j$ for $i\neq j$ because otherwise $m_i-m_j\in J$ and $p-b_i(m_i-m_j)\in I\setminus J$
would have fewer terms than $p$. Since the maps $t\mapsto e^{\mu_1 t},\ldots,t\mapsto e^{\mu_d t}$
are linearly independent, it follows that $b_1=\cdots=b_r=0$ which is contradiction.
Thus we must have have $I=J$. It is immediate to see that $J$ is generated by all polynomials $x^a-x^b$
such that $\lambda\cdot(a-b)=0$, thus it suffices to compute $L=\set{a\in\Z^d:\lambda\cdot a=0}$,
the set of additive relations of $\lambda$. Notice that $L$ is an additive subgroup of $\Z^d$ and
as such must be finitely generated. An upper bound on the size of the elements of a basis of $L$
can be found in \cite{Mas88} and therefore we can compute a basis for $L$ and thus $J$, $I$ and $\Zcl{\mathcal{O}}_A$.
\end{proof}

A corollary of this result is that we can compute $\Zcl{\mathcal{O}_A}$ by separating
the diagonal and nilpotent parts of $A$. The fact that this closure is computable is
not new, a similar result exists in the litterature, although without proof \cite{CE05}.

\begin{proposition}\label{prop:closure_lin_diff_eq}
There is an algorithm that given $A\in\Qbar^{d\times d}$, computes $\Zcl{\set{e^{At}:t\in\Rp}}$.
\end{proposition}
\begin{proof}
We can write $A=P(D+N)P^{-1}$ where $P$ is invertible, $D$ is diagonal, $N$
is nilpotent, and $D$ and $N$ commute. Notice that $\mathcal{O}_D$ only consists of diagonal matrices
and $\mathcal{O}_N$ of unipotent matrices. Since $\mathcal{O}_A$ is a commutative group,
and $D$ and $N$ commute, we have that $\mathcal{O}_A=P(\mathcal{O}_D\cdot\mathcal{O}_N)P^{-1}$,
and thus $\Zcl{\mathcal{O}_A}=P\Zcl{\Zcl{\mathcal{O}_D}\cdot\Zcl{\mathcal{O}_N}}P^{-1}$. All the
operations in this equation are effective thus it suffices to compute $\Zcl{\mathcal{O}_D}$ as
explained above, and $\Zcl{\mathcal{O}_N}$. But since $N$ is nilpotent, $q_n(t):=e^{Nt}$ is really a polynomial
in $Nt$ and thus in $t$. It follows that $\Zcl{\mathcal{O}_N}=\Zcl{q_n(\R)}=\Zcl{q_n(\Zcl{\R})}=\Zcl{q_n(\Cx)}$
which we know how to compute.
\end{proof}

Next we consider the class of \emph{switching systems}, that is,
hybrid systems in which every pair of locations is connected by an
edge and in which variables are no updated on discrete edges.  It is
known that reachability is undecidable even for this restricted
class~\cite{OPPW16} of systems.  However, as we show below, building
on Proposition~\ref{prop:closure_lin_diff_eq}, one can compute
computing strongest algebraic invariants in this setting.

\RestyleAlgo{ruled}
\begin{procedure}
\Input{$A_1,\ldots,A_k\in \Qbar^{d\times d}$}
$H:=\{I_d\}$ \; \label{alg:noreset:initH} 
$G_1:=\Zcl{ \{ e^{A_1t}:t\geq 0\}} ; \ldots ; G_k:=\Zcl{ \{ e^{A_1t}:t\geq 0\}}$\; \label{alg:noreset:inits}
\Repeat{$H_{old}=H$}{ \label{alg:noreset:loop:start}
    $H_{old}:=H$\;
    \For{$i \in \{1,\ldots,k\}$}{
        $H:=\Zcl{H\cdot  G_i}$ \; \label{alg:noreset:newH2}
    }
}
\label{alg:noreset:loop:end}
\Output{$H$}
\caption{Semigroup-Closure($A_1,\ldots,A_k$)}
\label{alg:noreset}
\end{procedure}

\begin{proposition}\label{prop:switched}
  Algorithm~\ref{alg:noreset} (shown below) terminates and outputs the
  Zariski closure of the sub-semigroup of $\mathrm{GL}_d(\Cx)$
  generated by $\{ e^{A_1t} ,\ldots, e^{A_kt}: t\geq 0\}$.
\end{proposition}
\begin{proof} 
  First note that the effectiveness of Line~\ref{alg:noreset:inits}
  relies on Proposition~\ref{prop:closure_lin_diff_eq}.

  Now we argue that $H$ in the algorithm is always an irreducible
  variety.  For this it suffices to show that if
  $X \subseteq \mathrm{GL}_d(\Cx)$ is an irreducible variety then so
  is $Y:=\Zcl{X \cdot \Zcl{\{e^{At} : t\geq 0\}}}$ for any matrix
  $A \in \Cx^{d\times d}$. First observe that if $X$ and $Z$ are
  irreducible sets then so is $\Zcl{X\cdot Z}$, thus is it enough to
  show that $G=\Zcl{S}$ is irreducible, where $S=\{e^{At} : t\geq 0\}$.
  But $S$ is a semigroup, thus $G$ is a group. This makes $G$ a linear
  algebraic group, which is therefore irreducible if and only if it
  is (Zariski-)connected. But $G$ being the closure of $S$ means it is
  enough to show that $S$ is Zariski-connected, and hence enough to show
  that it is Euclidean-connected. The latter is trivial since every element
  of $S$ is path-connected to $I_d$.

Now a strictly increasing chain $H_1\subseteq H_2\subseteq \cdots$ of
irreducible sub-varieties of $\mathrm{GL}_d(\Cx)$ has length at most
the dimension of $\mathrm{GL}_d(\Cx)$, which is $d^2$.  Thus
Algorithm~\ref{alg:noreset} terminates after at most $d^2$ iterations
of the outer loop.  It is clear that the terminating value of the
algorithm is the Zariski closure of the sub-semigroup of
$\mathrm{GL}_d(\Cx)$ generated by
$\{ e^{A_1t} ,\ldots, e^{A_kt}: t\geq 0\}$.
\end{proof}

Now consider a switching system $\mathcal{A}=(Q,A,E,T)$.  Let $G$ be
the Zariski closure of the semigroup generated by the matrices
$e^{A_qt}$, for $q\in Q$ and $t\geq 0$, which can be computed by
Proposition~\ref{prop:switched}.  Then $V(\mathcal{A})=\{V_q:q\in
Q\}$, the real Zariski closure of the collecting semantics of
$\mathcal{A}$, is such that $V_q = \Zcl{G\cdot X}\cap
\mathrm{GL}_d(\mathbb{R})$ for every $q\in Q$, where $X=\bigcup_{q\in
  Q} T_q$.  But then $V(\mathcal{A})$ is computable from $G$ and $T$.

\begin{theorem}
  Given a switching system $\mathcal{A}$, one can compute
  $V(\mathcal{A})$.
\end{theorem}

\section{Reducing Continuous Dynamics to Constructible Discrete Dynamics}
\label{sec:const-discretisation}
\begin{lemma}\label{lem:least_closed_fixpoint}
Let $\mathcal{A}$ be an automaton, then $V(\mathcal{A})$ is the least fixpoint
of the map $X\mapsto\Zcl{\Phi_\mathcal{A}(X)}$.
\end{lemma}
\begin{proof}
Let $S$ denotes the collecting semantics of $\mathcal{A}$, then
\begin{align*}
V(\mathcal{A})
    &=\Zcl{S}&&\text{by definition of $V(\mathcal{A})$}\\
    &=\Zcl{\Phi_\mathcal{A}(S)}&&\text{by definition of $S$}\\
    &=\Zcl{\Phi_\mathcal{A}(\Zcl{S})}&&\text{by Zariski-continuity of $\Phi_\mathcal{A}$}\\
    &=\Zcl{\Phi_\mathcal{A}(V(\mathcal{A}))}
\end{align*}
thus $V(\mathcal{A})$ is indeed a fixed-point.
Conversely, let $X$ be such that $X=\Zcl{\Phi_\mathcal{A}(X)}$. Then $X$ is closed and clearly
$\Phi_\mathcal{A}(X)\subseteq\Zcl{\Phi_\mathcal{A}(X)}=X$ so it is a pre-fixpoint of $\Phi_\mathcal{A}$.
By virtue of $S$ being the least (pre-)fixpoint of $\Phi_\mathcal{A}$ we must have $S\subseteq X$.
But then $\Zcl{S}\subseteq X$ \ie{} $V(\mathcal{A})\subseteq X$.
\end{proof}

\begin{proposition}\label{prop:lin_hybrid_to_constr_aff}
  Given a linear hybrid automaton $\mathcal{A}$, one can compute a
  constructible affine program $\mathcal{A}'$ that has the same
  algebraic invariants, \ie{} $V(\mathcal{A})=V(\mathcal{A}')$.
\end{proposition}
\begin{proof}
The idea is to replace each continuous dynamics $x'=A_qx$ by a closed (and thus constructible)
set of discrete transitions: $\Zcl{\set{e^{A_qt}:t\in\Rp}}$, which we know how to compute thanks
to Proposition~\ref{prop:closure_lin_diff_eq}
. Graphically:
\begin{center}
\begin{tikzpicture}[
        mynode/.style={fill=blue!30!white,circle,draw},
        mynode final/.style={mynode,double, double distance=0.5mm},
        myarrow/.style={->,thick},
        scale=0.7
    ]
    \node[mynode] at (0,0) (n1) {$x'=Ax$};
    \path (n1) ++(20:-2) edge[myarrow] (n1);
    \path (n1) ++(0:-2) edge[myarrow] (n1);
    \path (n1) ++(-20:-2) edge[myarrow] (n1);
    \path (n1) edge[myarrow] ++(-20:2);
    \path (n1) edge[myarrow] ++(0:2);
    \path (n1) edge[myarrow] ++(20:2);
    \path (3,0) edge[very thick,decorate, decoration={snake,amplitude=0.5ex,post=lineto,post length=1ex},-{Straight Barb[width=2ex]}] ++(1,0);
    \begin{scope}[shift={(7,0)}]
    \node[mynode] at (0,0) (n1) {$x'=0$};
    \path (n1) ++(20:-2) edge[myarrow] (n1);
    \path (n1) ++(0:-2) edge[myarrow] (n1);
    \path (n1) ++(-20:-2) edge[myarrow] (n1);
    \path (n1) edge[myarrow] ++(-20:2);
    \path (n1) edge[myarrow] ++(0:2);
    \path (n1) edge[myarrow] ++(20:2);
    \path (n1) edge[myarrow,loop,in=80,out=100,min distance=5ex] node[midway,above=-1ex] {$\Zcl{\set{e^{At}:t\in\Rp}}$} (n1);
    \end{scope}
\end{tikzpicture}
\end{center}
Formally, let $\mathcal{A}=(Q,A,E,T)$ be a linear hybrid automaton. Define the constructible
affine program $\mathcal{A}'=(Q,A',E',T)$ where $A'_q=0$ for all $q\in\Q$ and
\[E'=E\cup\set{(q,X,q):q\in\Q,X\in\Zcl{\mathcal{O}_{A_q}}}\]
where $\mathcal{O}_A:=\set{e^{At}:t\in\Rp}$. Note that it is constructible because $\Zcl{\mathcal{O}_{A_q}}$
is closed (and thus constructible). We will now relate $\Phi_\mathcal{A}$ and $\Phi_{\mathcal{A}'}$:
let $X\in\mathcal{P}(\Cx^d)^Q$ and observe that
\begin{align*}
\Phi_\mathcal{A}(X)_q
    &=T_q\cup(\mathcal{O}_{A_q}\cdot X_q)\cup\bigcup_{(p,B,q)\in E}BX_q\\
    &\subseteq T_q\cup(\Zcl{\mathcal{O}_{A_q}}\cdot X_q)\cup\bigcup_{(p,B,q)\in E}BX_q\\
    &=\Phi_{\mathcal{A}'}(X)_q.
\end{align*}
Let $S^\mathcal{A}$ (resp. $S^{\mathcal{A}'}$) denote the collecting semantics of $\mathcal{A}$ (resp. $\mathcal{A}'$),
that is $S^\mathcal{A}$ is the least fixpoint of $\Phi_\mathcal{A}$.
It follows that
\[\Phi_\mathcal{A}(S^\mathcal{A'})\subseteq\Phi_{\mathcal{A}'}(S^\mathcal{A'})=S^\mathcal{A'},\]
\ie{} $S^\mathcal{A'}$ is a pre-fixpoint of $\Phi_\mathcal{A}$ and thus it must be the case that
$S^\mathcal{A}\subseteq S^\mathcal{A'}$. It immediately follows that $V(\mathcal{A})\subseteq V(\mathcal{A}')$.
Conversely, we also have that
\begin{align*}
\Zcl{\Phi_\mathcal{A}(S)_q}
    &=\Zcl{T_q\cup(\mathcal{O}_{A_q}\cdot S_q)\cup\bigcup_{(p,B,q)\in E}BS_q}\\
    &=\Zcl{T_q\cup\Zcl{\mathcal{O}_{A_q}\cdot S_q}\cup\bigcup_{(p,B,q)\in E}BS_q}\\
    &=\Zcl{T_q\cup(\Zcl{\mathcal{O}_{A_q}}\cdot S_q)\cup\bigcup_{(p,B,q)\in E}BS_q}\\
    &=\Zcl{\Phi_{\mathcal{A}'}(S)_q}
\end{align*}
and similarly for $\mathcal{A}'$.
It follows that
\[V(\mathcal{A})_q
    =\Zcl{\Phi_\mathcal{A}(S^\mathcal{A})_q}
    =\Zcl{\Phi_{\mathcal{A}'}(S^\mathcal{A})_q}
    =\Zcl{\Phi_{\mathcal{A}'}(\Zcl{S^\mathcal{A})_q}}=\Zcl{\Phi_{\mathcal{A}'}(V(\mathcal{A}))_q}.
\]
Thus $V(\mathcal{A})$ is a fixpoint of $X\mapsto\Zcl{\Phi_{\mathcal{A}'}(X)}$. But $V(\mathcal{A}')$
is the least fixpoint of this map by Lemma~\ref{lem:least_closed_fixpoint} and thus $V(\mathcal{A}')\subseteq V(\mathcal{A})$.
\end{proof}

\section{Reducing Continuous Dynamics to Finite Discrete Dynamics}\label{sec:fin-discretisation}

In the previous section, we saw that given a linear hybrid automaton
one can compute a constructible affine program with the same set of
algebraic invariants.  In this section we show how to compute a
\emph{finite} affine program with the same set of algebraic
invariants.  The idea is to replace the continuous evolution of the
variables in each location of a hybrid automaton with a finite set of
discrete transitions.  Graphically this corresponds to rewriting the
automaton as follows:
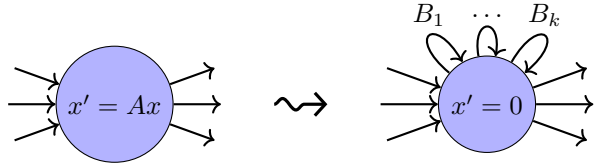
\begin{center}
\begin{tikzpicture}[
        mynode/.style={fill=blue!30!white,circle,draw},
        mynode final/.style={mynode,double, double distance=0.5mm},
        myarrow/.style={->,thick},
        scale=0.7
    ]
    \node[mynode] at (0,0) (n1) {$x'=Ax$};
    \path (n1) ++(20:-2) edge[myarrow] (n1);
    \path (n1) ++(0:-2) edge[myarrow] (n1);
    \path (n1) ++(-20:-2) edge[myarrow] (n1);
    \path (n1) edge[myarrow] ++(-20:2);
    \path (n1) edge[myarrow] ++(0:2);
    \path (n1) edge[myarrow] ++(20:2);
    \path (3,0) edge[very thick,decorate, decoration={snake,amplitude=0.5ex,post=lineto,post length=1ex},-{Straight Barb[width=2ex]}] ++(1,0);
    \begin{scope}[shift={(7,0)}]
    \node[mynode] at (0,0) (n1) {$x'=0$};
    \path (n1) ++(20:-2) edge[myarrow] (n1);
    \path (n1) ++(0:-2) edge[myarrow] (n1);
    \path (n1) ++(-20:-2) edge[myarrow] (n1);
    \path (n1) edge[myarrow] ++(-20:2);
    \path (n1) edge[myarrow] ++(0:2);
    \path (n1) edge[myarrow] ++(20:2);
    \path (n1) edge[myarrow,loop,in=120,out=140,min distance=7ex] node[midway,above] {$B_1$} (n1);
    \path (n1) edge[myarrow,loop,in=80,out=100,min distance=5ex] node[midway,above] {$\ldots$} (n1);
    \path (n1) edge[myarrow,loop,in=40,out=60,min distance=7ex] node[midway,above] {$B_k$} (n1);
    \end{scope}
\end{tikzpicture}
\end{center}

Mathematically, the task is as follows: \emph{Given
  $A\in\Qbar^{d\times d}$, find $B_1,\ldots,B_k\in\Qbar^{d\times d}$
  such that $\Zcl{\set{e^{At}:t\in\R}}=\Zcl{\Gen{B_1,\ldots,B_k}}$.}
There is a conceptually simple approach to this problem: namely for
every matrix $A\in \mathbb{A}^{d\times d}$ and rational number $\tau$
we have $\Zcl{\set{e^{At}:t\in\R}}=\Zcl{\Gen{e^{A\tau}}}$ (see
Section~\ref{sec:time-disc}).  But this does not fulfil our
desiderata, since it is not possible in general to find $\tau\in\R$
such that $e^{A\tau}$ has exclusively algebraic entries.  Nevertheless
given $A\in \mathbb{A}^{d\times d}$ it is possible to find
$B\in\mathbb{A}^{d\times d}$ such that
$\Zcl{\set{e^{At}:t\in\R}}=\Zcl{\Gen{B}}$.  The idea is to construct
$B$ such that there is a correspondence between the set of additive
relations satisfied by the eigenvalues of $A$ and the multiplicative
relations satisfied by the eigenvalues of $B$.

\begin{proposition}\label{prop:synth}
Let $a_1,\ldots,a_d\in\mathbb{C}$ be algebraic numbers.  Then we can
compute rational numbers  $\lambda_1,\ldots,\lambda_d$
such that $a_1n_1+\cdots+a_dn_d=0$ iff $\lambda_1^{n_1} \cdots
\lambda_d^{n_d}=1$ for all $n_1,\ldots,n_d \in \mathbb{Z}$.
\end{proposition}
\begin{proof}
  Let $s$ be the dimension of the $\mathbb{Q}$-vector space spanned by
  $a_1,\ldots,a_d$.  By computing a basis over $\mathbb{Q}$ for the number
  field $\mathbb{Q}(a_1,\ldots,a_d)$ and the respective rational coordinates 
of $a_1,\ldots,a_d$ with respect to this basis, we obtain an $s\times d$ integer matrix
$A$ such that for every integer vector
$\boldsymbol{x}=(n_1,\ldots,n_d) \in \mathbb{Z}^d$ we have $n_1a_1+
\cdots +n_da_d = 0$ iff $A\boldsymbol{x}=0$.  

Now write $A=PBQ$, where
$B$ is an $s\times d$ matrix in Smith normal form and $P,Q$ are
unimodular square matrices.  Since $B$ has rank $s$ it has the form
$B=\begin{pmatrix}D & 0\end{pmatrix}$ for $D$ an $s\times s$ diagonal
matrix of full rank.

We define positive integers $\mu_1,\ldots,\mu_d$ as follows.
Choose $\mu_{1},\ldots,\mu_{s}$ to be the first $s$ prime numbers 
and let
$\mu_{s+1}=\ldots=\mu_d=1$.  Write $Q = (q_{ij})$ and define
$\lambda_i = \mu_1^{q_{1i}} \cdots \mu_d^{q_{di}}$ for $i\in
\{1,\ldots,d\}$.  

Then for all $n_1,\ldots,n_d \in \mathbb{Z}$ we have
\begin{eqnarray*}
\lambda_1^{n_1} \cdots \lambda_d^{n_d} =1& \Leftrightarrow &
\mu_1^{(Q\boldsymbol{x})_1} \cdots \mu_d^{(Q\boldsymbol{x})_d} = 1\\
& \Leftrightarrow & (Q\boldsymbol{x})_1=0,\ldots,(Q\boldsymbol{x})_s=0 \\
& \Leftrightarrow & BQ\boldsymbol{x} = 0 \quad\text{(since $B=\begin{pmatrix} D& 0\end{pmatrix}$)}\\
& \Leftrightarrow & PBQ\boldsymbol{x} = 0\quad\text{(since $P$ is invertible)}\\
& \Leftrightarrow & A\boldsymbol{x} = 0\\
& \Leftrightarrow & a_1n_1+\cdots+a_dn_d=0 \, .
\end{eqnarray*}
\end{proof}

\begin{corollary}\label{corl:synth}
  Let $D$ be a $d\times d$ diagonal matrix with algebraic entries.
  Then there exists a diagonal matrix $D'$, of the same dimension and
  with rational entries, such that $\Zcl{\Gen{e^D}} = \Zcl{\Gen{D'}}$.
\end{corollary}
\begin{proof}
Write $D=\mathrm{diag}(a_1,\ldots,a_d)$ and let rational numbers
$\lambda_1,\ldots,\lambda_d$ be chosen as in Proposition~\ref{prop:synth}, i.e.,
such that $a_1n_1+\cdots+a_dn_d=0$ iff $\lambda_1^{n_1} \cdots
\lambda_d^{n_d}=1$ for all $n_1,\ldots,n_d \in \mathbb{Z}$.  
Define $D'=\mathrm{diag}(\lambda_1,\ldots,\lambda_d)$.
By Lemma~\ref{lem:ideal_lin_diff_eq_diag}, the ideal of the variety $\Zcl{\Gen{e^D}}$
is generated by $z^n-z^m$ such that $(n_1-m_1)a_1+\cdots+(n_d-m_d)a_d=0$.
On the other hand, it follows from~\cite[Lemma 6]{DerksenJK05} that
the ideal of the variety $\Zcl{\Gen{D'}}$ is generated
by $z^n-z^m$ such that $\lambda_1^{n_1-m_1}\cdots\lambda_d^{n_d-m_d}=1$.
But by construction the additive relations of $a$ are the same
as the multiplicative relations of $\lambda$, therefore the ideals are the same.
It follows $\Zcl{\Gen{e^D}} = \Zcl{\Gen{D'}}$.
\end{proof}

\begin{proposition}\label{prop:disc2}
Let $A\in \mathbb{Q}^{d\times d}$ be a rational matrix.
Then there exists an algebraic matrix $B$ such that
$\overline{\langle B \rangle} = \overline{\langle e^A \rangle} = \overline{\{e^{At}:t\in \mathbb{R} \}}$. 
\end{proposition}
\begin{proof}
  Let $P$ be an invertible matrix such that $A = P^{-1}(D+N)P$ with
  $D=\mathrm{diag}(a_1,\ldots,a_d)$ diagonal and $N$ a nilpotent
  Jordan matrix.  By Corollary~\ref{corl:synth} there exists a
  rational diagonal matrix $D'$ such that
  $\overline{\langle D'\rangle} = \overline{\langle e^D\rangle}$.  We
  now define $B := P^{-1}(D'e^{N})P$ where
  $D'=\mathrm{diag}(\lambda_1,\ldots,\lambda_d)$.  Note that $e^{N}$
  is a matrix of rational numbers.  Then we have:
  \[
    \overline{\langle e^A \rangle}
    = P^{-1} \overline{\langle e^{D}e^{N}\rangle } P
    = P^{-1} \overline{\overline{\langle e^{D} \rangle} \cdot
                \overline{\langle e^{N} \rangle}} P 
    = P^{-1} \overline{\overline{\langle D' \rangle} \cdot
                \overline{\langle e^{N} \rangle}} P
    = P^{-1} \overline{\langle D'e^{N} \rangle } P\\
    = \overline{\langle B \rangle} \, .
  \]
\end{proof}

\begin{proposition}
  Given a linear hybrid automaton $\mathcal{A}$, one can compute a
  finite affine program $\mathcal{A}'$ that the same algebraic
  invariants, i.e., $V(\mathcal{A})=V(\mathcal{A}')$.
\end{proposition}
\begin{proof}
Suppose that $\mathcal{A}=(Q,A,E,T)$. We define
$\mathcal{A}'=(Q,A',E',T)$, where
$A'_q=0$ for all $q\in Q$ and
\[ E'=E\cup \{ (q,B_q,q) : q \in Q\} \, ,\]
with $B_q$ is an algebraic matrix such that
$\Zcl{\Gen{e^{A_q}}} = \Zcl{\Gen{B_q}}$ for all $q \in Q$.
The existence of the matrices $B_q$ is guaranteed by
Proposition~\ref{prop:disc2}.  In other words, we 
obtain $\mathcal{A}'$ from $\mathcal{A}$ by setting the derivative of all variables to $0$
in every location and by adding a compensatory selfloop edge to every location.

The reasoning that $V(\mathcal{A})=V(\mathcal{A}')$ is entirely
analogous to that in the proof of
Proposition~\ref{prop:lin_hybrid_to_constr_aff}.
\end{proof}


\section{Switching Systems with Guards}
\label{sec:guards}
In this section we consider linear hybrid automata with no discrete
updates on the variables but with equality guards on the discrete mode
changes.  For such systems there is a smallest algebraic inductive
invariant, which can be obtained as the (location-wise) intersection
of the family of all algebraic inductive invariants.  However, as we
show in this section, this invariant is no longer computable.  In
other words, in the presence of equality guards the analog of
Theorem~\ref{thm:main} fails.  (As an aside we remark that the
discrete mode changes no longer induce Zariski continuous maps on
configurations if there are equality guards.  Hence we cannot
necessarily recover the smallest algebraic inductive invariant as the
Zariski closure of the collecting semantics).

\begin{theorem}\label{th:invariant_with_guards_undec}
  There is no algorithm that computes the strongest algebraic
  inductive invariant for the class of switching systems with equality
  guards.
\end{theorem}

\begin{proof}[Proof Sketch (see full proof in appendix)] The idea is to simulate
a 2-counter machine in such a way that if the machine has an infinite
run then the strongest invariant has dimension $2$, and otherwise it
has dimension $1$. Since the dimension of an algebraic set can be effectively
determined; this concludes the sketch.
\end{proof}

\bibliographystyle{plainurl}
\bibliography{literature}

\appendix

\appendix
\section{Time Discretisation}
\label{sec:time-disc}






\begin{proposition}\label{prop:disc}
For a rational matrix $A \in \mathbb{Q}^{d \times d}$ we have
\[ \overline{\{ e^{At} : t \in \mathbb{R} \}} =
   \overline{\langle e^A \rangle } \, . \] 
\end{proposition}
\begin{proof}
Suppose that $A$ is diagonalisable---say $A=U^{-1}DU$ for some invertible matrix $U$ and $D=\mathrm{diag}(a_1,\ldots,a_d)$.
It suffices to prove that 
 $\overline{\{ e^{Dt} : t \in \mathbb{R} \}} =
   \overline{\langle e^D \rangle }$.

  Consider a
  multiplicative relationship among the eigenvalues of $e^D$---say
  $(e^{a_1})^{n_1} \cdots (e^{a_d})^{n_d}=1$, where
  $n_1,\ldots,n_d \in \mathbb{Z}$.
Then $a_1n_1+\cdots + a_dn_d \in (2\pi i)\mathbb{Z}$.
But since $a_1,\ldots,a_d$ are algebraic numbers, we must in fact have
$a_dn_d+\cdots + a_dn_d=0$.  
It follows that $a_1tn_1+\cdots +a_dtn_d =0$ for all $t\in\mathbb{R}$
and hence $(e^{a_1t})^{n_1} \cdots (e^{a_dt})^{n_d}=1$ for all
$t\in \mathbb{R}$, i.e., the same multiplicative relation also holds among the eigenvalues of $e^{Dt}$.

Since the ideal of all polynomial relations satisfied by
$\langle e^{D} \rangle$ is generated by the multiplicative relations
satisfied by the eigenvalues of $e^D$, we have that for any
$t\in \mathbb{R}$, matrix $e^{Dt}$ satisfies all polynomial relations
satisfied by $\langle e^D \rangle$.  This proves the proposition in
case $A$ is diagonalisable.

Next, suppose that $A$ is nilpotent.  The fact that
$\overline{\{e^{At}:t\in\mathbb{R} \}} = \overline{\langle e^{A} \rangle}$
is already shown in Section 3.3 of Derksen, Jeandel, and Koiran.  

The general case can by handled by reduction to the
diagonalisable and nilpotent cases as in Proposition~\ref{prop:closure_lin_diff_eq}.
\end{proof}

Proposition~\ref{prop:disc} crucially relies on the fact that $\pi$
does not appear in the description of $A$.  Indeed, consider the case
that $A=\begin{pmatrix}2\pi i\end{pmatrix}\in\Cx^{1\times1}$. Then
$\set{e^{At}:t\in\R}=\set{z\in\Cx:|z|=1}$ is the unit circle. However
$\set{e^{An}:n\in\Z}=\set{1}$ is a singleton.  Such an example is
possible because the map $z\in\Cx\mapsto e^{Az}$ is not
Zariski-continuous in general.

\section{Proof of Theorem~\ref{th:invariant_with_guards_undec}}

Recall that a non-deterministic 2-counter machine $M$ consists of two
counters $C$ and $D$ and a list of $n$ instructions.  Each instruction
increments one of the counters, decrements one of the counters, or
tests one of the counters for zero.  After executing a counter update
or a successful test, the machine proceeds nondeterministically to one
of two specified instructions.  The machine halts if it executes a
test instruction whose condition is false.  Given an instruction $i$,
if $j$ is one of the two possible successors of $i$ then we call the
pair $(i,j)$ a \emph{transition} of $M$.  Initially $M$ starts with both
counters zero and instruction $1$ is the first to be executed.  A
configuration of $M$ is a triple consisting of the current instruction
and the current counter values.  The problem of whether such a machine
$M$ can reach infinitely many configurations from its initial
configuration is undecidable.

Corresponding to such a 2-counter machine $M$ we define a linear
hybrid automaton $\mathcal{A}=(Q,A,E,q)$ in dimension $3$.  We think
of $\mathcal{A}$ as having continuous variables $c,d,t$, where $c$ and
$d$ respectively correspond to the counters of $M$.  Each variable has
constant derivative in each location, which is zero unless otherwise
specified.  For each instruction $i$ of $M$ we postulate a location
$q_i$ of $\mathcal{A}$ and for each transition $(i,j)$ of $M$ we
postulate a location $q_{i,j}$ of $\mathcal{A}$.  
Variable $t$ has slope $1$ in each location $q_i$ and slope $-1$ in
each location $q_{i,j}$.  For every transition$(i,j)$ of $M$,
automaton $\mathcal{A}$ has an edge from $q_i$ to $q_{i,j}$ with guard
$t=1$ and an edge from $q_{i,j}$ to $q_j$ with guard $t=0$.
Intuitively, if an execution of $\mathcal{A}$ correctly simulates a
run of $M$ then $\mathcal{A}$ spends one time unit in each location,
alternating between locations $q_i$ that correspond to instructions of
$M$ and locations $q_{i,j}$ that correspond to transitions of $M$.

Suppose that the $i$-th instruction of $M$ performs an incrementation
$C:=C+1$.  Then variable $C$ has slope $1$ in location $q_i$.
Likewise if the $i$-th instruction if $C:=C-1$, then variable $c$ has
slope $-1$ in location $q_i$.  If the $i$-th instruction of $M$ is the
zero test $C=0$, then the edge from location $q_i$ to $q_{i,j}$ in $A$
has guard $c=0$.  There are corresponding constructions for counter
operations and tests on counter $D$.

This completes the description of $\mathcal{A}$.
We now claim that:
\begin{enumerate}
\item If $M$ can only reach finitely many configurations from the
  initial configuration then $V(\mathcal{A})=\{ V_q : q \in Q\}$, the
  Zariski closure of the collecting semantics, is an
  inductive invariant that has dimension one.
\item If infinitely many configuration are reachable from the initial configuration of $M$
then the smallest inductive invariant has dimension strictly greater than one.
\end{enumerate}

To prove the claim, note that for each reachable configuration
$(i,z_1,z_2)$ of $M$, the collecting semantics
$\Phi(\mathcal{A})_{q_i}$ contains a half-line $L$ containing the
point $(z_1,z_2,0)$, whose direction is determined by the slopes of
the respective variables of $\mathcal{A}$ in location $q_i$.  The
Zariski closure of $L$ is the affine hull of $L$, i.e., the
corresponding full line containing $L$.  Crucially, the points added
to $L$ to obtain the full line are all predecessors of $L$ under the
flow relation of $\mathcal{A}$.  In particular the Zariski closure is
inductive: it remains closed under the transition relation of
$\mathcal{A}$.  In particular, if $M$ can only reach finitely many
configurations, then the Zariski closure of the collecting semantics
consists of finitely many lines in each location (and so has dimension
one) and is moreover an inductive invariant.

Suppose $M$ can reach infinitely many configurations.  Since any
algebraic inductive invariant must in particular contain the Zariski
closure of the collecting semantics, it follow that any algebraic
inductive invariant must contain infinitely many lines in some
location and thus must have dimension strictly greater than one.

Since the dimension of an algebraic set can be effectively determined,
we conclude that it is not possible to compute the smallest algebraic
invariant of a linear hybrid automaton with equality guards (even with
a no discrete updates of the variables).

\end{document}